%% file: main.tex
\documentclass[a4paper,UKenglish,cleveref, autoref, thm-restate,authorcolumns]{lipics-v2021}
\usepackage{tikz}

\pdfoutput=1 
\hideLIPIcs  

\title{Sorting Balls and Water: Equivalence and Computational Complexity} 

\author{Takehiro Ito}{Graduate School of Information Sciences, Tohoku University, Japan}{takehiro@tohoku.ac.jp}{https://orcid.org/0000-0002-9912-6898}{Partially supported by JSPS KAKENHI Grant Numbers JP19K11814 and JP20H05793, Japan.}
\author{Jun Kawahara}{Kyoto University, Japan.}{jkawahara@i.kyoto-u.ac.jp}{}{}
\author{Shin-ichi Minato}{Kyoto University, Japan \and \url{https://www.lab2.kuis.kyoto-u.ac.jp/minato/}}{minato@i.kyoto-u.ac.jp}{https://orcid.org/0000-0002-1397-1020}{}
\author{Yota Otachi}{Nagoya University, Japan \and \url{https://www.math.mi.i.nagoya-u.ac.jp/~otachi/} }{otachi@nagoya-u.jp}{https://orcid.org/0000-0002-0087-853X}
{JSPS KAKENHI Grant Numbers
  JP18K11168, 
  JP18K11169, 
  JP20H05793, 
  JP21K11752. 
}
\author{Toshiki Saitoh}{Kyushu Institute of Technology, Japan.}{toshikis@ai.kyutech.ac.jp}{}{}
\author{Akira Suzuki}{Graduate School of Information Sciences, Tohoku University, Japan \and \url{http://www.ecei.tohoku.ac.jp/alg/suzuki/}}{akira@tohoku.ac.jp}{https://orcid.org/0000-0002-5212-0202}{Partially supported by JSPS KAKENHI Grant Numbers JP20K11666 and JP20H05794, Japan.}
\author{Ryuhei Uehara}{Japan Advanced Institute of Science and Technology, Japan \and \url{http://www.jaist.ac.jp/~uehara}}{uehara@jaist.ac.jp}{https://orcid.org/0000-0003-0895-3765}{JSPS KAKENHI Grant Numbers
 JP20H05961, 
 JP20H05964, 
 JP20K11673. 
}
\author{Takeaki Uno}{National Institute of Informatics, Japan.}{uno@nii.jp}{}{}
\author{Katsuhisa Yamanaka}{Iwate University, Japan \and \url{http://www.kono.cis.iwate-u.ac.jp/~yamanaka/}}{yamanaka@iwate-u.ac.jp}{https://orcid.org/0000-0002-4333-8680}{Partially supported by JSPS KAKENHI Grant Numbers 19K11812, Japan.}
\author{Ryo Yoshinaka}{Graduate School of Information Sciences, Tohoku University, Japan.}{ryoshinaka@tohoku.ac.jp}{https://orcid.org/0000-0002-5175-465X}{}

\authorrunning{T.~Ito et al.} 

\Copyright{Takehiro Ito, Jun Kawahara, Shin-ichi Minato, Yota Otachi,
Toshiki Saitoh, Akira Suzuki, Ryuhei Uehara, Takeaki Uno, Katsuhisa Yamanaka, and Ryo Yoshinaka
} 

\ccsdesc[100]{Mathematics of computing $\rightarrow$
Discrete mathematics $\rightarrow$
Combinatorics $\rightarrow$
Combinatorial algorithms} 

\keywords{Ball sort puzzle, 
recreational mathematics,
sorting pairs in bins, 
water sort puzzle} 

\category{} 

\relatedversion{} 

\funding{This research is partially supported by JSPS Kakenhi Grant Number 18H04091. 
}

\nolinenumbers 

\EventEditors{John Q. Open and Joan R. Access}
\EventNoEds{2}
\EventLongTitle{42nd Conference on Very Important Topics (CVIT 2016)}
\EventShortTitle{CVIT 2016}
\EventAcronym{CVIT}
\EventYear{2016}
\EventDate{December 24--27, 2016}
\EventLocation{Little Whinging, United Kingdom}
\EventLogo{}
\SeriesVolume{42}
\ArticleNo{23}

\newcommand{\emptyseq}{\varepsilon}
\newcommand{\mcal}[1]{\mathcal{#1}}
\newcommand{\ballmove}{\rightarrow}
\newcommand{\watermove}{\Rightarrow}
\newcommand{\tbmove}[1]{\stackrel{#1}{\Rrightarrow}}
\newcommand{\lrceil}[1]{\left\lceil #1 \right\rceil}
\newcommand{\topcolor}[1]{\mathrm{TC}(#1)}
\newcommand{\topborder}[1]{\mathrm{TB}_{#1}}

\newcommand{\pname}{\mathsf}
\newcommand{\BSP}{\pname{BSP}}
\newcommand{\WSP}{\pname{WSP}}

\usepackage{color}
\definecolor{darkred}{rgb}{0.75,0,0}
\definecolor{darkgreen}{rgb}{0,0,1}

\newcommand\msize[1]{\left|#1\right|}
\newcommand{\ceil}[1]{\lceil #1 \rceil}
\newcommand{\floor}[1]{\lfloor #1 \rfloor}

\newcommand{\ptitle}[1]{\medskip\noindent\hspace*{2mm}\textsc{\underline{#1}}} 
\newenvironment{listing}[1]{%
        \begin{list}{*}{%
                 \settowidth{\labelwidth}{#1}%
                 \setlength{\leftmargin}{\labelwidth}%
                 \advance \leftmargin by 12pt
                   \setlength{\itemsep}{0pt}%
                   \setlength{\parsep}{0pt}%
                   \setlength{\topsep}{0pt}%
                   \setlength{\parskip}{0pt}%
}%
}{%
\end{list}}

\begin{document}

\maketitle

\begin{abstract}
Various forms of sorting problems have been studied over the years.
Recently, two kinds of sorting puzzle apps are popularized.
In these puzzles, we are given a set of bins filled with colored units, balls or water, and some empty bins.
These puzzles allow us to move colored units from a bin to another 
when the colors involved match in some way or the target bin is empty.
The goal of these puzzles is to sort all the color units in order.
We investigate computational complexities of these puzzles.
We first show that these two puzzles are essentially the same from the viewpoint of solvability.
That is, an instance is sortable by ball-moves if and only if it is sortable by water-moves.
We also show that every yes-instance has a solution of polynomial length, which implies that
these puzzles belong to NP\@.
We then show that these puzzles are NP-complete.
For some special cases, we give polynomial-time algorithms.
We finally consider the number of empty bins sufficient for making all instances solvable
and give non-trivial upper and lower bounds in terms of the number of filled bins and the capacity of bins.
\end{abstract}

\section{Introduction}
\label{sec:intro}
\input{intro}


\section{Equivalence of balls and water}
\label{sec:equiv}
\input{equiv}

\section{NP-completeness}
\label{sec:np}
\input{np}

\section{Polynomial-time algorithms when $h=2$ and $|C| = n$}
\label{sec:h2}
\input{h2}

\section{The number of sufficient bins}
\label{sec:bins}
\input{bins}

\section{Concluding remarks}
\label{sec:conc}
\input{conc}

\bibliography{main}
\end{document}

%% file: intro.tex

The \emph{ball sort puzzle} and the \emph{water sort puzzle} 
are popularized recently via smartphone apps.\footnote{These sort puzzles are released by 
several different developers. As far as the authors checked,
it seems that the first one is released by IEC Global Pty Ltd in January, 2020.} 
Both puzzles involve bins filled with some colored units (balls or water),
and the goal is to somehow sort them.
The most significant feature of these puzzles is that each bin works as a stack. 
That is, the items in a bin have to follow the ``last-in first-out'' (LIFO) rule.

\begin{figure}\centering
\includegraphics[width=0.9\linewidth]{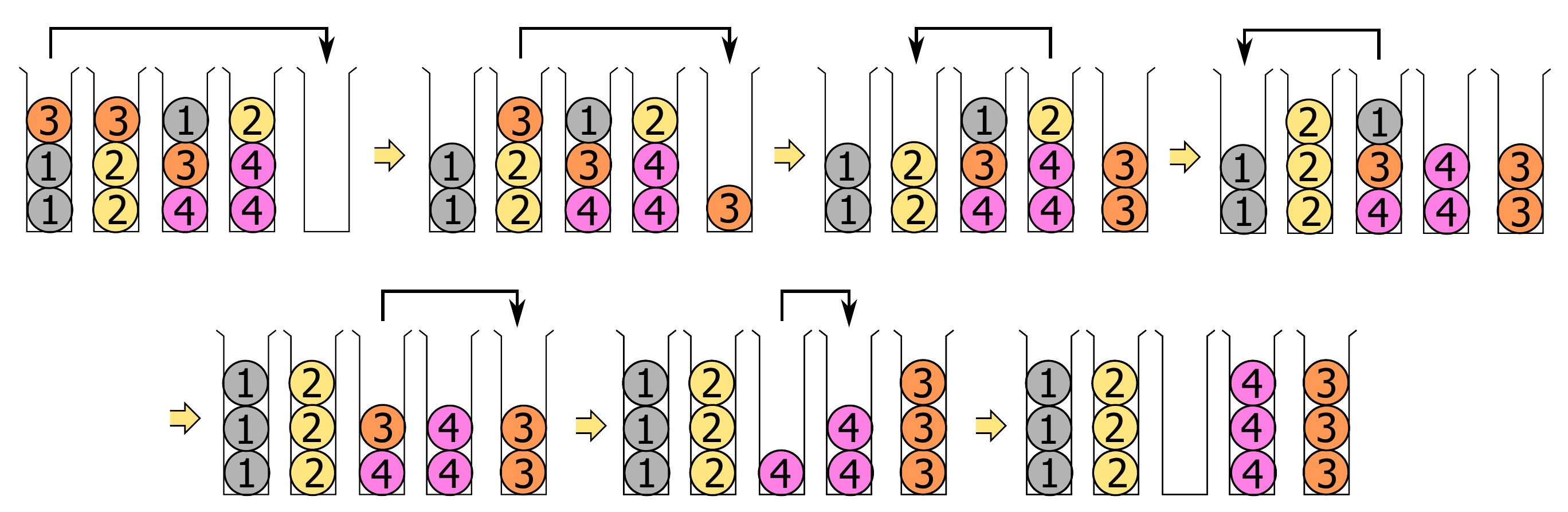}
\caption{An example of the ball sort puzzle.} 
\label{fig:ball1}
\end{figure}

In the ball sort puzzle, we are given $h n$ colored balls in $n$ bins of capacity $h$ and $k$ additional empty bins.
For a given (unsorted) initial configuration, the goal of this puzzle is to arrange the balls in order; 
that is, to make each bin either empty or full with balls of the same color.
(The ordering of bins does not matter in this puzzle.)
The rule of this puzzle is simple: 
(0) Each bin works like a stack, that is, we can pick up the top ball in the bin.
(1) We can move the top ball of a bin to the top of another bin if the second bin is empty or 
it is not full and its top ball before the move and the moved ball have the same color.
An example with $h = 3$, $n = 4$, and $k=1$ is given in \figurename~\ref{fig:ball1}.

\begin{figure}\centering
\includegraphics[width=0.9\linewidth]{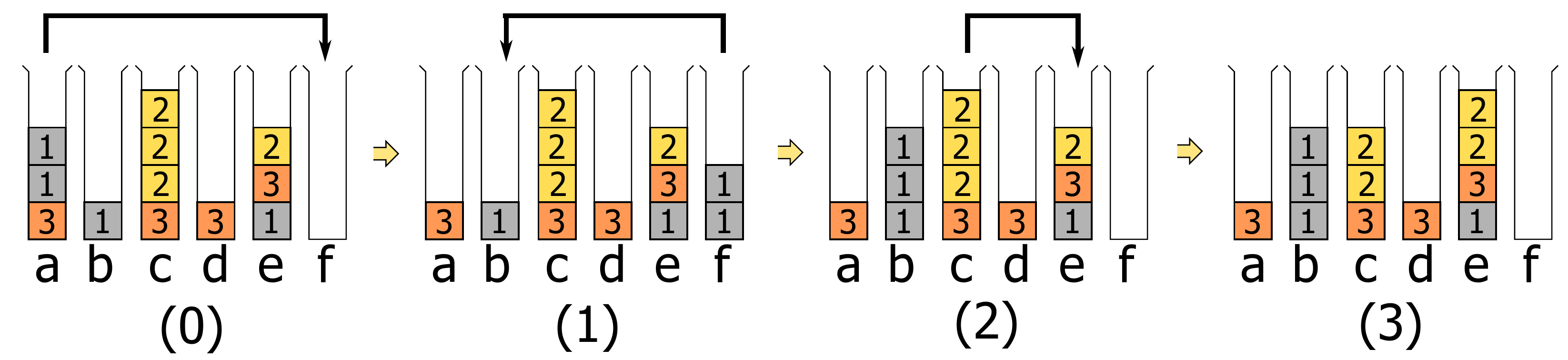}
\caption{Rules of the water sort puzzle.
From the initial configuration (0), we obtain
configuration (1) by moving two units of water of label $1$ from \texttt{a} to \texttt{f},
configuration (2) by moving two units of water of label $1$ from \texttt{a} to \texttt{b}, and
configuration (3) by moving one unit of water of label $2$ from \texttt{c} to \texttt{e}.} 
\label{fig:water1}
\end{figure}

The water sort puzzle is similar to the ball sort puzzle.
Each ball is replaced by colored water of a unit volume in the water sort puzzle.
In the water sort puzzle, the rules (0) and (1) are the same as the ball sort puzzle except one liquid property: 
Colored water units are merged when they have the same color and they are consecutive in a bin.
When we pick up a source bin and move the top water unit(s) to a target bin, 
the quantity of the colored water on the top of the bin to be moved varies according to the following conditions (\figurename~\ref{fig:water1}).
If the target bin has enough margin, all the water of the same color moves to the target bin.
On the other hand, a part of the water of the same color moves up to the limit of the target bin 
if the target bin does not have enough margin.


In this paper, we investigate computational complexities of the ball and water sort puzzles.
We are given $n+k$ bins of capacity $h$.
In an initial configuration, $n$ bins are full (i.e., filled with $h$ units) and $k$ bins are empty,
where each unit has a color in the color set $C$.
Our task is to fill $n$ bins monochromatically, that is, we need to fill each of them with $h$ units of the same color.
We define $\BSP$ and $\WSP$ as the problems of deciding whether a given initial configuration
can be reconfigured to a sorted configuration by a sequence of ball moves and water moves, respectively.
We assume that instances are encoded in a reasonable way, which takes $\Theta(nh \log |C| + \log k)$ bits.
Without loss of generality, we assume that each color $c \in C$ occurs exactly $hj$ times for some positive integer $j$ in any instance
since otherwise the instance is a trivial no-instance. (This implies that $|C|$ is at most $n$.)

\subsection{Our results}

We first prove that the problems $\BSP$ and $\WSP$ are actually equivalent.
By their definitions, a yes-instance of $\WSP$ is a yes-instance of $\BSP$ as well.
Thus our technical contribution here is to prove the opposite direction.
As a byproduct, we show that a yes-instance of the problems 
admits a reconfiguration sequence of length polynomial in $n+h$ as a yes-certificate.
This implies that the problems belong to NP\@.
We also show that $\BSP$ and $\WSP$ are solvable in time $O(h^{n})$.

We next show that $\BSP$ and $\WSP$ are indeed NP-complete.
By slightly modifying this proof, we also show that even for some kind of \emph{trivial} yes-instances of $\WSP$,
it is NP-complete to find a shortest reconfiguration to a sorted configuration.

We show that if the capacity $h$ is $2$ and the number of colors is $n$,
then $\BSP$ and $\WSP$ are polynomial-time solvable.
In this case, we present an algorithm that finds shortest reconfiguration sequences for yes-instances.

We also consider the following question:
how many empty bins do we need to ensure that all initial configurations can reach a sorted configuration?
Observe that $\BSP$ and $\WSP$ are trivial if $k \ge hn$.
It is not difficult to see that $k \ge n$ is also sufficient by using an idea based on bucket sort.
By improving this idea, we show that $k \ge \ceil{\frac{h-1}{h}n}$ is sufficient for all instances.
We also show that some instances need $k \ge \floor{\frac{19}{64}\min\{n,h\}}$ empty bins.

\subsection{Related studies}

Various forms of sorting problems have been studied over the years.
For example, sorting by reversals is well investigated in the contexts of sorting network \cite[Sect.\ 5.2.2]{Knu98}, 
pancake sort \cite{GatesPapadimitriou1979}, and ladder lotteries \cite{YamanakaNakanoMatsuiUeharaNakada2010}.
There is another extension of sorting problem from the viewpoint of recreational mathematics defined as follows.
For each $i \in \{1,\dots,n\}$, there are $h$ balls labeled $i$.
These $hn$ balls are given in $n$ bins in which each bin is of capacity $h$.
In one step, we are allowed to swap any pair of balls.
Then how many swaps do we need to sort the balls?
This sorting problem was first posed by Peter Winker for the case $h=2$ in 2004~\cite{Winkler2004},
and an almost optimal algorithm for that case was given by Ito et al.~in 2012~\cite{ItoTeruyamaYoshida2012}.

The ball/water sorting puzzles are interesting also from the viewpoints of not only just puzzles but also combinatorial reconfiguration.
Recently, the \emph{reconfiguration problems} attracted the attention in theoretical computer science (see, e.g., \cite{IDHPSUU}).
These problems arise when we need to find a step-by-step transformation between two feasible solutions of 
a problem such that all intermediate results are also feasible and each step abides by a fixed reconfiguration rule, 
that is, an adjacency relation defined on feasible solutions of the original problem.
In this sense, the ball/water sort puzzles can be seen as typical implementations of the framework of reconfiguration problems, while their reconfiguration rules are non-standard.
In most reconfiguration problems, representative reconfiguration rules are \emph{reversible}; 
that is, if a feasible solution $A$ can be reconfigured to another feasible solution $B$,
then $B$ can be reconfigured to $A$ as well (see e.g., the puzzles in \cite{HearnDemaine2009}).
In the ball sort puzzle, we can reverse the last move
if we have two or more balls of the same color at the top of the source bin, or 
there is only one ball in the source bin. 
Otherwise, we cannot reverse the last move.
That is, some moves are reversible while some moves are one-way in these puzzles.
(In \figurename~\ref{fig:ball1}, only the last move is reversible.)
In the water sort puzzle, basically, 
we cannot separate units of water of the same color once we merge them,
however, there are some exceptional cases.
An example is given in \figurename~\ref{fig:water1};
the move from configuration (2) to configuration (3) is reversible, 
but the other moves are not.

%% file: equiv.tex
\newcommand{\binpic}[5]{%
	\tikzset{R/.style={fill=red!20}}
	\tikzset{G/.style={fill=green!20}}
	\tikzset{B/.style={fill=blue!20}}
	\tikzset{K/.style={fill=black!20}}
	\tikzset{O/.style={fill=none}}
	\begin{scope}[xshift=#1cm]
	\draw[thick] (-0.1,#5+0.3) -- (0,#5+0.2) -- (0,#5) -- (0,0) -- (1,0) -- (1,#5) -- (1,#5+0.2) -- (1.1,#5+0.3); 
	\foreach \dy/\cl [%
 		remember = \newy (initially 0),
 		evaluate = \dy as \oldy using {\newy},
 		evaluate = \dy as \newy using {\oldy+\dy},
 		count = \k,
 		] in {#3} {
 		\filldraw[\cl] (0,\newy) -- (0,\oldy) -- (1,\oldy) -- (1,\newy) -- cycle;
	\foreach \i in {\oldy,...,\newy}
		[\draw[dotted] (0,\i) -- (1,\i);]
 	}
	\draw[ultra thick] (0,#4) -- (1,#4);
	\end{scope}
}

This section presents fundamental properties of $\BSP$ and $\WSP$, from which we will conclude that
\begin{itemize}
\item a configuration is a yes-instance of $\BSP$ if and only if it is a yes-instance of $\WSP$, and
\item $\BSP$ and $\WSP$ both belong to $\mathrm{NP}$.
\end{itemize}

\subsection{Notation used in this section}

Let $B$ and $C$ be finite sets of \emph{bins} and \emph{colors}, respectively.
The set of sequences of colors of length at most $h$ is denoted by $C^{\le h}$.
A sequence of $c \in C$ of length $l$ is called \emph{monochrome} and denoted by $c^l$.
The empty sequence is denoted as $\emptyseq$, which is, in our terminology, also called monochrome.
A \emph{configuration} is a map $S \colon B \to C^{\le h}$.
An instance of $\BSP$ and $\WSP$ is a configuration $S$ such that $|S(b)| \in \{h,0\}$ for all $b \in B$.
A bin $b$ is \emph{sorted} under $S$ if $S(b) \in \{\emptyseq, c^h\}$ for some $c \in C$. 
A configuration $S$ is \emph{sorted} or \emph{goal} if all bins are sorted under $S$.
Sometimes we identify $b \in B$ with $S(b)$ when $S$ is clear from the context.

The $i$th element of a sequence $\alpha$ is denoted by $\alpha[i]$ for $1 \le i \le |\alpha|$.
The \emph{top color} of a color sequence $\alpha$ is $\topcolor{\alpha} = \alpha[|\alpha|]$ if $\alpha$ is not empty.
In case $\alpha$ is empty, $\topcolor{\alpha}$ is defined to be $\emptyseq$. 
A \emph{border} in $\alpha$ is a non-negative integer $i$ such that either $1 \le i < |\alpha|$ and $\alpha[{i}] \neq \alpha[{i+1}]$ or $i=0$, and
the set of borders of $\alpha$ is denoted by $\mcal{D}(\alpha)$.
We count non-trivial borders under $S$ as $D(S) = \sum_{b\in B}|\mcal{D}(S(b))\setminus\{0\}|$.
For example, in \figurename~\ref{fig:water1}, the values of $D$ are $4,3,3,3$ in (0), (1), (2), (3), respectively.

We now turn to define a move of the sort puzzles with the terminologies introduced above.
Intuitively, we pick up two bins $b_1$ and $b_2$ from $B$,
and pour the top item(s) from $b_1$ to $b_2$ if $b_2$ is empty or the top item of $b_2$ has the same color.
The quantity of the items are different in the cases of balls and water.
In the ball case, a unit is moved if possible. On the other hand, two or more units of water of the same color
move at once until all units are moved or $b_2$ becomes full. We define the move more precisely below.

A \emph{ball-move} can modify a configuration $S$ into $S'$, denoted as $S \ballmove S'$, if there are $b_1,b_2\in B$ and $c \in C$ such that
\begin{itemize}
  \item $S(b_1)= S'(b_1) \cdot c$, $\topcolor{S(b_2)}\in \{c,\emptyseq\}$, $S'(b_2)= S(b_2) \cdot c$, and 
  \item $S(b)=S'(b)$ for all $b\in B\setminus\{b_1,b_2\}$.
\end{itemize}

A \emph{water-move} can modify $S$ into $S'$, denoted as $S \watermove S'$, if there are $m \ge 1$, $b_1,b_2\in B$
and $c \in C$ such that 
\begin{itemize}
  \item $S(b_1)= S'(b_1)\cdot c^m$, $\topcolor{S(b_2)}\in \{c,\emptyseq\}$, and $S'(b_2)= S(b_2) \cdot c^m$, 
  \item either $\topcolor{S'(b_1)}\neq c$ or $\msize{S'(b_2)}=h$, and
  \item $S(b)=S'(b)$ for all $b\in B\setminus\{b_1,b_2\}$.
\end{itemize}

The reflexive and transitive closures of $\ballmove$ and $\watermove$ are denoted as $\ballmove^*$ and $\watermove^*$, respectively.
Clearly any single water-move can be simulated by some number of ball-moves.
Thus we have ${\watermove^*} \subseteq {\ballmove^*}$.

\subsection{Fundamental Theorems of $\BSP$ and $\WSP$}

Hereafter we fix an arbitrary initial configuration $S_0$.
Some notions introduced below are relative to $S_0$, though it may not be explicit.
The \emph{top-border table} of a configuration $S$ is a map $\topborder{S} \colon B \to \{0,1,\dots,h-1\}$ defined as $\topborder{S}(b) = \max \mcal{D}(S(b))$.
 
The main object of this section is to observe that top-border tables play an essential role in analyzing the solvability of an instance.
Note that $\topborder{S}(b)=0$ if and only if $b$ is monochrome.
Once we have reached a configuration with $\topborder{S}(b)=0$ for all $b \in B$,
we can achieve a goal configuration by trivial moves.
Figure~\ref{fig:topborder} illustrates some notions introduced in this section.

By the definition of moves, one can easily observe the following properties.
\begin{observation}
\label{obs:basic}
	If $S \ballmove^* T$, the following holds for all bins $b \in B$:
	\begin{enumerate}
        \item Moves monotonically remove borders from top to bottom:\label{obs:monotone}
        \[
        \mcal{D}(T(b)) = \{\, i \in \mcal{D}(S(b)) \mid i \le \topborder{T}(b)\,\}\,;
        \]
		\item The contents below the top borders of $T$ do not change:\label{obs:stable}
		\[
		    T(b)[i] = S(b)[i] \text{ for all } i \le \topborder{T}(b)\,;
		\]
        \item If $b$ is not monochrome in $T$, the colors of all the units above the border in $T$ are the color of the unit just above that border in $S$:\label{obs:topcolor}
		\[
		\topcolor{T(b)} = S(b)[\topborder{T}(b) +1]
		\text{ if\/ $\topborder{T}(b) \neq 0 $.}
		\]
	\end{enumerate}
\end{observation}
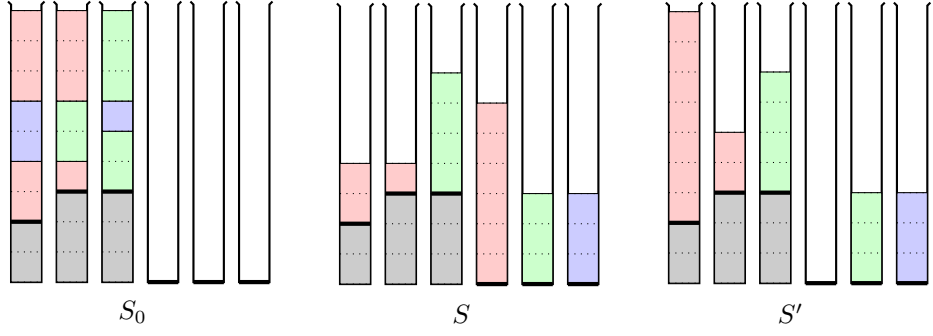
\begin{figure}\centering
	\begin{tikzpicture}[scale=0.4]
		\binpic{0.0}{0}{2/K,2/R,2/B,3/R}{2}{9}
		\binpic{1.5}{0}{3/K,1/R,2/G,3/R}{3}{9}
		\binpic{3.0}{0}{3/K,2/G,1/B,3/G}{3}{9}
		\binpic{4.5}{0}{0/O}{0}{9}
		\binpic{6.0}{0}{0/O}{0}{9}
		\binpic{7.5}{0}{0/O}{0}{9}
		\draw(4,-1) node{$S_0$};
	\end{tikzpicture}
	\quad\quad
	\begin{tikzpicture}[scale=0.4]
		\binpic{0.0}{0}{2/K,2/R}{2}{9}
		\binpic{1.5}{0}{3/K,1/R}{3}{9}
		\binpic{3.0}{0}{3/K,4/G}{3}{9}
		\binpic{4.5}{0}{6/R}{0}{9}
		\binpic{6.0}{0}{3/G}{0}{9}
		\binpic{7.5}{0}{3/B}{0}{9}
		\draw(4,-1) node{$S$};
	\end{tikzpicture}
	\quad\quad
	\begin{tikzpicture}[scale=0.4]
		\binpic{0.0}{0}{2/K,7/R}{2}{9}
		\binpic{1.5}{0}{3/K,2/R}{3}{9}
		\binpic{3.0}{0}{3/K,4/G}{3}{9}
		\binpic{4.5}{0}{0/O}{0}{9}
		\binpic{6.0}{0}{3/G}{0}{9}
		\binpic{7.5}{0}{3/B}{0}{9}
		\draw(4,-1) node{$S'$};
	\end{tikzpicture}
	\caption{\label{fig:topborder} The configuration $S$ can be obtained from the initial configuration $S_0$ and can be modified into the tight one $S'$. The top-borders $\topborder{S}$ of $S$ are shown with thick lines. This figure does not care the units below those borders. All of $S_0$, $S$, and $S'$ have $\mcal{F}_\textrm{red}(\topborder{S})=9$ red units, $\mcal{F}_\textrm{green}(\topborder{S})=7$ green units, and $\mcal{F}_\textrm{blue}(\topborder{S})=3$ blue units above the borders in total. The colors red, green, and blue require $\mcal{M}_{\textrm{red}}(\topborder{S})=0$, $\mcal{M}_{\textrm{green}}(\topborder{S})=1$, and $\mcal{M}_{\textrm{blue}}(\topborder{S})=1$ monochrome bins, respectively, which coincide with the number of the monochrome bins of respective colors in $S'$.}
\end{figure}

We first give a necessary condition for $S_0 \ballmove^* S$.

Recall that throughout the game play, the total number of units of each color does not change, i.e.,
if $S_0 \ballmove^* S$, it holds for every color $c \in C$,
\[
	\sum_{b \in B}|\{\, i \mid 1 \le i \le h \text{ and } S_0(b)[i]=c\,\}|
	= \sum_{b \in B}|\{\, i \mid 1 \le i \le |S(b)| \text{ and } S(b)[i]=c\,\}|\,.
\]
Since the units under the top border of each bin have not been moved (Observation~\ref{obs:basic}-\ref{obs:stable}), 
the total numbers of units of each color $c$ above the borders of $\topborder{S}$ coincide in $S_0$ and $S$.
We count those units of color $c$ as
\[
	\mcal{F}_{c}(\topborder{S}) = \sum_{b \in B} | \{\, i \mid \topborder{S}(b) < i \le h \text{ and } S_0(b)[i] = c \,\} |
\,.\]
Those units may have been moved, but units of color $c$ can go only to empty bins or bins whose top color is $c$ in $S$.
Let us partition the bins into $|C|+1$ groups with respect to $S$, where 
\begin{itemize}
	\item $B_\varepsilon(\topborder{S}) = \{\, b \in B \mid \topborder{S}(b)=0 \,\}$: monochrome bins,
	\item $B_c(\topborder{S}) = \{\, b \in B \setminus B_\varepsilon \mid S_0(b)[\topborder{S}(b)+1]=c\,\}$: non-monochrome bins with the top color $c$.
\end{itemize}
Note that if $b \notin B_\varepsilon(\topborder{S})$, the top color of $b$ in $S$ is $\topcolor{S(b)} = S(b)[\topborder{S}(b)+1] = S_0(b)[\topborder{S}(b)+1]$.
Since each bin $b \in B_c$ may have at most $h - \topborder{S}(b)$ units of color $c$ above the border $\topborder{S}(b)$,
 there are 
 \[
 	\mcal{G}_c(\topborder{S}) = \sum_{b \in B_c} (h-\topborder{S}(b))
 \]
 units of color $c$ can be on the top layer of non-monochrome bins.
Thus, we must have at least
\[
	\mcal{M}_c(\topborder{S}) = \max\left\{0,\, \lrceil{\frac{\mcal{F}_c(\topborder{S}) - \mcal{G}_c(\topborder{S})}{h}}\right\}
\]
monochrome bins in $B_\varepsilon$ which have some units of color $c$.
Therefore, it is necessary that
\begin{equation}\label{eq:good_configuration}
	\sum_{c \in C} \mcal{M}_c(\topborder{S}) \le | B_\varepsilon(\topborder{S}) |
\,.\end{equation}
We remark that the values of the functions introduced above, namely $\mcal{F}_c$, $B_c$, $\mcal{M}_c$, are all determined by the top-border table $\topborder{S}$. Those values coincide for $S$ and $S'$ if $\topborder{S}=\topborder{S'}$.
Let us say that $S$ is \emph{consistent with} $S_0$ if and only if $\topborder{S}$ satisfies Eq.~(\ref{eq:good_configuration}).
Moreover, $S$ is \emph{tight} if $S$ has exactly $\mcal{M}_c(\topborder{S})$ monochrome bins with at least one unit of every color $c \in C$.
\begin{lemma}\label{lem:tight}
Suppose $S$ is consistent with $S_0$.
There is a tight configuration $S'$ such that $\topborder{S}=\topborder{S'}$ and $S \watermove^* S'$.
\end{lemma}
\begin{proof}
	Recall that it is necessary that $S$ has at least $\mcal{M}_c(\topborder{S})$ monochrome bins of each color $c$.
	If $S$ has more than that, one can move all the units of any of the monochrome bins of color $c$ to other non-empty bins.
	The definition of $\mcal{M}_c$ guarantees that we can repeat this until we have just $\mcal{M}_c(\topborder{S})$ monochrome bins of color $c$.
\end{proof}

Of course $\topborder{S}$ being consistent with $S_0$ does not imply $S_0 \ballmove^* S$.
Suppose $S_0 \ballmove^* S \ballmove S'$ with $\topborder{S} \neq \topborder{S'}$, where $S'$ has one less border than $S$ by moving a unit above the top border of some bin $b$ to some other bins.
Note that $S(b)$ is not monochrome. 
During the move, the bin $b$ can keep units below its top border only.
Therefore, the number of necessary monochrome bins of each color $c$ for this move is
\begin{equation*}
	\mcal{M}_c^{b}(\topborder{S}) = \max\left\{0,\, \lrceil{\frac{\mcal{F}_c(\topborder{S}) - (\mcal{G}_c(\topborder{S}) - (h-\topborder{S}(b)))}{h}}\right\}
\end{equation*}
and it is necessary that 
\begin{equation*}
	\sum_{c \in C} \mcal{M}_c^{b}(\topborder{S}) \le | B_\varepsilon(\topborder{S}) |
\,.\end{equation*}
Those conditions depend on top-border configuration of $S$.
For two top-border tables $\tau$ and $\tau'$ and a bin $b \in B$, we write $\tau \tbmove{b} \tau'$ if
\begin{itemize}
	\item $\tau(a)=\tau'(a)$ for all $a \in B \setminus \{b\}$,
	\item $\tau(b) > 0$ and $\tau'(b) = \max \{\, i \in \mcal{D}(S_0(b)) \mid 0 \le i < \tau(b)\,\}$,
	\item $\sum_{c \in C} \mcal{M}_c^{b}(\tau) \le | B_\varepsilon(\tau) |$.
\end{itemize}
We write $\tau \tbmove{} \tau'$ if $\tau \tbmove{b} \tau'$ for some $b \in B$.
Figure~\ref{fig:topborder_nonmovable} may help intuitive understanding the above argument.

%

\begin{theorem}\label{thm:movable_condition}
	Suppose a configuration $S$ is consistent with $S_0$.
	Then, there is a configuration $T$ such that $S \watermove^* T$ if and only if $\topborder{S} \tbmove{}^* \topborder{T}$,
	where $\tbmove{}^*$ is the reflexive and transitive closure of $\tbmove{}$.
\end{theorem}
\begin{proof}
	Suppose $S \watermove T$ and $\topborder{S} \neq \topborder{T}$.
	Then $\topborder{S}$ and $\topborder{T}$ must satisfy the condition for $\topborder{S} \tbmove{} \topborder{T}$.
	
	Suppose $\topborder{S} \tbmove{} \topborder{T}$.
	We assume without loss of generality that $S$ is tight by Lemma~\ref{lem:tight}.
	The condition $\sum_{c \in C} \mcal{M}_c^{b}(\tau) \le | B_\varepsilon(\tau) |$ implies that units of color $c=S_0(b)[\topborder{S}(b)+1]$ 	in concern can be distributed to bins other than $b$.
	If $\mcal{M}_c^b(\topborder{S}) = \mcal{M}_c(\topborder{S})$, one can move the units in $b$ above $\topborder{T}(b)$ to other bins whose top color is $c$.
	If $\mcal{M}_c^b(\topborder{S}) > \mcal{M}_c(\topborder{S})$, one can move those units to an empty bin, which must be available in $S$, since $S$ is tight.
\end{proof}


\begin{figure}\centering
	\begin{tikzpicture}[scale=0.4]
		\binpic{0.0}{0}{6/K,2/R}{6}{8}
		\binpic{1.5}{0}{5/K,3/R}{5}{8}
		\binpic{3.0}{0}{4/K,4/R}{4}{8}
		\binpic{4.5}{0}{5/R}{0}{8}
		\draw(0.5,-0.6) node{$b_1$};
		\draw(2.0,-0.6) node{$b_2$};
		\draw(3.5,-0.6) node{$b_3$};
		\draw(5,-0.6) node{$b_4$};
		\draw(2.5,-2) node{$S_1$};
	\end{tikzpicture}
	\quad\quad
	\begin{tikzpicture}[scale=0.4]
		\binpic{0.0}{0}{6/K,1/R}{6}{8}
		\binpic{1.5}{0}{5/K,2/R}{5}{8}
		\binpic{3.0}{0}{4/K,3/R}{4}{8}
		\binpic{4.5}{0}{8/R}{0}{8}
		\draw(0.5,-0.6) node{$b_1$};
		\draw(2.0,-0.6) node{$b_2$};
		\draw(3.5,-0.6) node{$b_3$};
		\draw(5,-0.6) node{$b_4$};
		\draw(2.5,-2) node{$S_2$};
	\end{tikzpicture}
	\quad\quad
	\begin{tikzpicture}[scale=0.4]
		\binpic{0.0}{0}{6/K,1/R}{6}{8}
		\binpic{1.5}{0}{5/K,3/R}{5}{8}
		\binpic{3.0}{0}{4/K,3/R}{4}{8}
		\binpic{4.5}{0}{7/R}{0}{8}
		\draw(0.5,-0.6) node{$b_1$};
		\draw(2.0,-0.6) node{$b_2$};
		\draw(3.5,-0.6) node{$b_3$};
		\draw(5,-0.6) node{$b_4$};
		\draw(2.5,-2) node{$S_3$};
	\end{tikzpicture}
	\caption{\label{fig:topborder_nonmovable} All configurations $S_1$, $S_2$ and $S_3$ are tight and have the same top-border table $\tau=\topborder{S_1}=\topborder{S_2}=\topborder{S_3}$, where $\tau(b_1)=7$, $\tau(b_2)=5$, and $\tau(b_3)=4$.
	 In those configurations, one can move the units above the top border of any of $b_1$ or $b_2$, but it is impossible for $b_3$.
	 This solely depends on the top-border table $\tau$ and the total number $\mcal{F}_\textrm{red}(\tau)=16$ of red units above the top borders.
	 It is independent of how those red units are distributed over the bins in the current configuration.}
\end{figure}
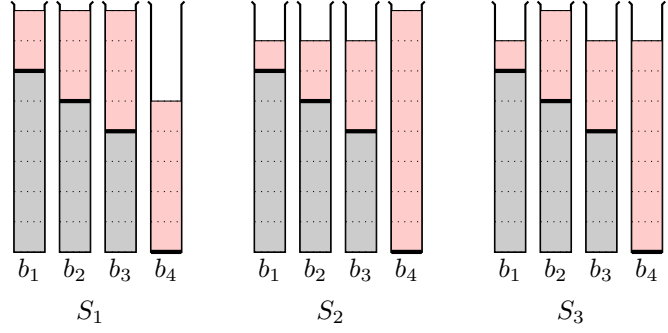
\begin{corollary}\label{cor:BSPvsWSP}
An initial configuration $S_0$ is a yes-instance of\/ $\BSP$ if and only if it is a yes-instance of\/ $\WSP$.
\end{corollary}
\begin{proof}
Since any water-move can be simulated by some ball-moves, it is enough to show the converse.
We claim that, for any configurations $S$ and $S'$ such that $\topborder{S} = \topborder{S'}$, if $S \ballmove T$, then there is $T'$ such that $S' \watermove^* T'$ and $\topborder{T} = \topborder{T'}$. 

By Theorem~\ref{thm:movable_condition}, either $\topborder{S}=\topborder{T}$ or $\topborder{S} \tbmove{} \topborder{T}$.
If $\topborder{S}=\topborder{T}$, then $T'=S'$ proves the claim.
Otherwise, the claim immediately follows from Theorem~\ref{thm:movable_condition}.
Thus, if $S_0 \ballmove^* T$ and $T$ is a goal configuration, then one can have $S_0 \watermove^* T'$ for some $T'$ with $\topborder{T} = \topborder{T'}$.
By modifying $T'$ tight by Lemma~\ref{lem:tight}, we get a goal configuration.
\end{proof}

If $\tau \tbmove{b} \tau'$, from the values $\mcal{F}_{c}(\tau)$ and $\mcal{G}_{c}(\tau)$, one can compute $\mcal{F}_c(\tau')$ and $\mcal{G}_c(\tau')$ in constant time, by preprocessing $S_0$, using the following equations:
\begin{align*}
	\mcal{F}_{c}(\tau') &= \begin{cases}
		\mcal{F}_c(\tau)+(\tau(b)-\tau'(b)) &\text{if $S_0(b)[\tau'(b)+1] = c$,}
\\		\mcal{F}_c(\tau) &\text{otherwise,}
	\end{cases}
\\
	\mcal{G}_c(\tau') &= \begin{cases}
		\mcal{G}_c(\tau) - (h-\tau(b))  & \text{if $S_0(b)[\tau(b)+1] = c$,}
\\		\mcal{G}_c(\tau) + (h-\tau'(b))  & \text{if $S_0(b)[\tau'(b)+1] = c$ and $\tau'(b) > 0$,}
\\		\mcal{G}_c(\tau) & \text{otherwise.}
	\end{cases}
\end{align*}

\begin{corollary}\label{cor:NP}
	$\BSP$ and $\WSP$ belong to $\mathrm{NP}$.
\end{corollary}
\begin{proof}
	By Theorem~\ref{thm:movable_condition}, $S_0$ is a yes-instance if and only if there is a bin sequence $(b_1,\dots,b_m)$ with
	$m = D(S_0)$ ($= \sum_{b\in B}|\mcal{D}(S_{0}(b))\setminus\{0\}|$) 
	that admits a top-border table sequence $(\tau_0,\dots,\tau_m)$ such that $\tau_0=\topborder{S_0}$ and $\tau_{i-1} \tbmove{b_i} \tau_i$ for all $i$.
	Each $\tau_i$ is uniquely determined by $\tau_{i-1}$ and $b_i$ and one can verify $\tau_{i-1} \tbmove{b_i} \tau_i$ in constant time, by maintaining the values of $\mcal{F}_c$ and $\mcal{G}_c$ (or just $\mcal{F}_c-\mcal{G}_c$).
\end{proof}
	A solution for an instance is essentially an order of borders of the instance to remove.
\begin{corollary}
	$\BSP$ and $\WSP$ can be solved in $O(h^{n})$ time. 
\end{corollary}
\begin{proof}
	There are at most $\prod_{b \in B} |\mcal{D}(S_0(b))| \le h^n$ distinct top-border tables. 
\end{proof}


\begin{corollary}\label{cor:shortest-path-ball}
For every yes-instance of\/ $\BSP$, there exists a sequence of ball-moves to a goal configuration of length at most $(2h-1)hn$.
\end{corollary}
\begin{proof}
In $\BSP$, $h-1$ ball-moves are enough to eliminate a border from a tight configuration, if there is a way to eliminate the border.
To make the obtained configuration tight takes at most $h$ moves, since the obtained configuration may have at most one monochrome bin to empty when the previous configuration is tight.
Since there can be at most $(h-1)n$ borders to eliminate, every yes-instance has a solution consisting of at most $(2h-1)hn$ ball-moves.
\end{proof}

\begin{corollary}\label{cor:shortest-path-water}
For every yes-instance of\/ $\WSP$, there exists a sequence of water-moves to a goal configuration of length at most $2(h-1)nl$ where $l=\min\{h,n\}$.
\end{corollary}
\begin{proof}
In $\WSP$, $l=\min\{h,n\}$ water-moves are enough to eliminate a border from a tight configuration, if there is a way to eliminate the border.
To make the obtained configuration tight, it takes again $l$ moves.
To see this, observe that the obtained configuration may have at most one monochrome bin to empty as in the ball case.
If one water-move does not empty the source monochrome bin, it must make the target bin full.
Since one cannot make more than $n$ full bins, to empty a monochrome bin takes at most $l$ water-moves.
Therefore, every yes-instance has a solution consisting of at most $2(h-1)nl$ water-moves.
\end{proof}

%% file: np.tex
In this section, we show that $\BSP$ and $\WSP$ are NP-complete even with two colors.
By \cref{cor:BSPvsWSP,cor:NP}, it suffices to show that $\WSP$ with two colors is NP-hard.
By slightly modifying the proof, we also show that given a trivial yes-instance of $\WSP$,
it is NP-complete to find a shortest sequence of water-moves to a sorted configuration,
where an instance is trivial in the sense that it contains many (say $hn$) empty bins.

\begin{theorem}
\label{th:NP-h}
$\BSP$ and $\WSP$ are NP-complete even with two colors.
\end{theorem}

\begin{proof}
As mentioned above, it suffices to show that $\WSP$ with two colors is NP-hard.
We prove it by a reduction from the following problem \textsc{3-Partition},
which is known to be NP-complete even if $B$ is bounded from above by some polynomial in $m$~\cite{GJ79}.

\ptitle{3-Partition}
\begin{listing}{aaa}
\item[{\bf Input:}] Positive integers $a_1, a_2, a_3, \dots, a_{3m}$ such that
	     $\sum_{i=1}^{3m}a_i=mB$ for some positive integer $B$ and $B/4<a_i<B/2$ for $1\le i\le 3m$.
\item[{\bf Question:}] Is there a partition of $\{1,2,\dots,3m\}$ into 
	     $m$ subsets $A_1,A_2,\dots,A_m$ such that $\sum_{i\in A_j} a_i=B$ for $1\le j\le m$?
\end{listing}
\medskip

In the reduction, we use two colors red and blue.
A non-empty bin is \emph{red} (\emph{blue}) if it contains red (blue, resp.) units only.
A non-empty bin is \emph{red-blue} if its top units are red and the other units are blue.
From an instance $\langle a_{1}, \dots, a_{3m} \rangle$ of \textsc{3-Partition},
we define an instance $S$ of $\WSP$ as follows (see \cref{fig:water}):
\begin{itemize}
  \item the capacity of each bin is $B$;
  
  \item for each $i \in [3m]$, it contains a full red-blue bin $b_{i}$ with $a_{i}$ red units and $B-a_{i}$ blue units;
  
  \item it contains one empty bin.
\end{itemize}

We show that $S$ is a yes-instance of $\WSP$
if and only if $\langle a_{1}, \dots, a_{3m} \rangle$ is a yes-instance of \textsc{3-Partition}.
\begin{figure}[h]
  \centering
  \includegraphics[scale=0.75]{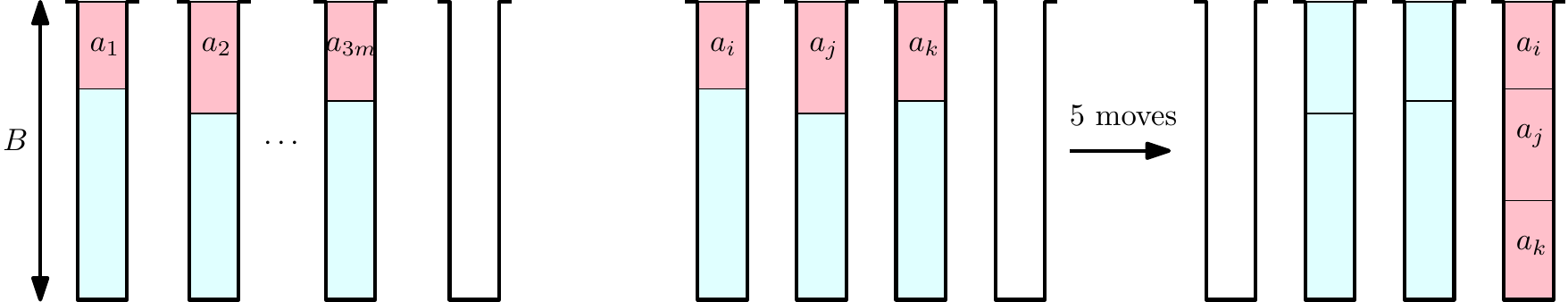}
  \caption{The reduction from \textsc{3-Partition} to $\WSP$.}
  \label{fig:water}
\end{figure}

To show the if direction, assume that $\langle a_{1}, \dots, a_{3m} \rangle$ is a yes-instance of \textsc{3-Partition}.
We can construct a sequence of water-moves from $S$ to a sorted configuration as follows.
Let $A_{1}, \dots, A_{m}$ be a partition of $\{1,\dots,3m\}$ such that $\sum_{i \in A_{j}} a_{i} = B$ for $1 \le j \le m$.
Using one empty bin, we can sort $b_{j_{1}}, b_{j_{2}}, b_{j_{3}}$ with $A_{j} = \{j_{1}, j_{2}, j_{3}\}$ as follows (see \cref{fig:water}).
We first move all red units to the empty bin. 
Since the number of red units is $a_{j_{1}} + a_{j_{2}} + a_{j_{3}} = B$, the bin becomes full.
Now $b_{j_{1}}$, $b_{j_{2}}$, and $b_{j_{3}}$ are blue bins 
containing $(B - a_{j_{1}}) + (B - a_{j_{2}}) + (B - a_{j_{3}}) = 2B$ blue units in total.
We move the units in $b_{j_{3}}$ to $b_{j_{1}}$ and $b_{j_{2}}$.
After that, $b_{j_{1}}$ and $b_{j_{2}}$ become full and $b_{j_{3}}$ is empty.
Using this new empty bin, we can continue and sort all bins.

To show the only-if direction, we prove the following slightly modified statement. 
(What we need is the case of $\rho = \beta = 0$.)
\begin{quote}
Let $\rho$ and $\beta$ be non-negative integers,
and $S_{0}$ be the instance of $\WSP$ obtained from $S$ by adding $\rho$ full red bins and $\beta$ full blue bins.
If $S_{0}$ is a yes-instance of the decision problem, 
then $\langle a_{1}, \dots, a_{3m} \rangle$ is a yes-instance of \textsc{3-Partition}.
\end{quote}

Assume that there is a sorting sequence $S_{0}, \dots, S_{\ell}$, where $S_{\ell}$ is a sorted configuration.
We use induction on $m$. (We need the dummy bins for the induction step.)
The case of $m = 1$ is trivial since all instances of \textsc{3-Partition} with $m=1$ are yes-instances.
Assume that $m \ge 2$ and that the statement holds for strictly smaller instances of \textsc{3-Partition}.

Observe that $S_{0}$ contains $3m + \rho + \beta + 1$ bins
and $(3m + \rho + \beta)B$ units of water ($(m + \rho)B$ red units and $(2m + \beta)B$ blue units).
Since the capacity of bins is $B$, every $S_{i}$ contains at most one empty bin, and if there is one, then the other bins are full.
Observe also that each bin in each $S_{i}$ is either red, blue, red-blue, or empty.

We say that a move \emph{opens} a bin $b_{i}$ if the move makes $b_{i}$ red units free for the first time.
Observe that each red unit in $b_{1}, \dots, b_{3m}$ has to move at most once since otherwise non-monochromatic bins will remain.
Thus, all $b_{1}, \dots, b_{3m}$ will eventually be opened.
Let $b_{p}$, $b_{q}$, and $b_{r}$ be the first three bins opened in the sorting sequence in this order.

Assume that $b_{r}$ is opened by the move from $S_{j-1}$ to $S_{j}$. Let $\alpha = a_{p} + a_{q} + a_{r}$.
\begin{claim}
\label{clm:S_j-conditions}
$S_{j}$ satisfies the following conditions.
\begin{enumerate}
  \item There are $\rho + 1$ red bins, and they contain at least $\rho B  + \alpha$ units.
  \item There are $\beta+3$ blue bins, and they contain $(\beta+3) B - \alpha$ units.
  \item There are $3m - 3$ red-blue bins, and they contain $(2 m-3) B + \alpha$ blue units and at most $m B - \alpha$ red units.
  \item There is no empty bin. (This follows for free from the other conditions.)
\end{enumerate}
\end{claim}
\begin{proof}
Observe that we cannot create a new red-blue bin by any sequence of moves.
Thus the number of red-blue bins is $3m-3$.
Since each red-blue bin $b_{i}$ contains exactly $B - a_{i}$ blue units and at most $a_{i}$ red units,
the red-blue bins contain exactly $(2 m-3) B + \alpha$ blue units
and at most $m B - \alpha$ red units in total.

The blue bins contain exactly $(\beta+3) B - \alpha$ units.
Since $b_{r}$ contains exactly $B - a_{r}$ units,
the other blue bins contain $(\beta+2) B - (a_{p} + a_{q})$ units,
and thus the number of blue bins is at least $\lceil ((\beta+2) B - (a_{p} + a_{q})) / B\rceil + 1 \ge \beta + 3$,
where the inequality holds as $a_{p} + a_{q} < B$.

Since the red bins contain at least $\rho B + \alpha$ units,
the number of red bins is at least $\rho + \lceil \alpha/B \rceil \ge \rho + 1$.
Recall that the total number of bins is $3m + \rho + \beta + 1$.
Thus, 
the number of red bins is exactly $\rho + 1$ and
the number of blue bins is exactly $\beta + 3$.
\end{proof}

\begin{claim}
$\alpha = B$.
\end{claim}
\begin{proof}
In $S_{j}$, the $\rho + 1$ red bins contain at least $\rho B + \alpha$ units. 
This implies that $\alpha \le B$.
Suppose to the contrary that $\alpha < B$. 
We show that under this assumption, the conditions in \cref{clm:S_j-conditions} cannot be violated by any sequence of moves.
This contradicts that the final state $S_{\ell}$ does not contain any red-blue bin.

Let $T$ be a configuration satisfying the conditions in \cref{clm:S_j-conditions}
and $T'$ be a configuration obtained from $T$ by one move.
It suffices to show that $T'$ still satisfies the conditions in \cref{clm:S_j-conditions}.
Assume that the move is from a bin $b$ to another bin $b'$.

First consider the case where we moved blue units.
In this case, $b$ and $b'$ have to be blue bins. 
Hence, to show that all properties are satisfied, it suffices to show that $b$ is not empty after the move.
Indeed, if $b$ becomes empty after the move, then 
$\beta + 2$ blue bins contain $(\beta+3) B - \alpha > (\beta+2)B$ units in total. 
This contradicts the capacity of bins.

Next consider the case where we moved red units.
The bins $b$ and $b'$ are red or red-blue.
If $b$ becomes empty (when it was red) or blue (when it was red-blue),
then the total number of red bins and red-blue bins
becomes $3m + \rho -3$ but they still contain $(3m+\rho-3)B + \alpha$ units
($(m+\rho)B$ red and $(2m-3)B + \alpha$ blue).
This contradicts the capacity of bins.
Thus, we know that the type of $b$ does not change by the move.
If the move is either from red to red, from red-blue to red-blue, or from red-blue to red,
then the all conditions are satisfied.
Assume that the move is from red to red-blue and that the red-blue bins contain more than $m B - \alpha$ red units after the move.
Now the $3m-3$ red-blue bins contain more than $(2m-3)B+\alpha + m B - \alpha = (3m-3)B$ units.
This again contradicts the capacity of bins.
\end{proof}

By the claims above and the capacity of bins, $S_{j}$ satisfies the following conditions.
\begin{enumerate}
  \item There are $\rho + 1$ full red bins.
  \item There are $\beta+3$ blue bins, and they contain $(\beta+2) B$ units.
  \item There are $3m - 3$ full red-blue bins, and they contain $(2m-2)B$ blue units and $(m-1)B$ red units.
  \item There is no empty bin.
\end{enumerate}
Observe that each red-blue bin $b_{i}$ in $S_{j}$ is not opened so far,
and thus contains $B - a_{i}$ blue units (and $a_{i}$ red units as it is full).

Let us take a look at the sorting sequence from $S_{j}$ to $S_{\ell}$.
Since there is no empty bin and all bins containing red units are full in $S_{j}$,
we need to move blue units in blue bins first.
Unless we make an empty bin, the situation does not change.
Let $S_{h}$ be the first configuration after $S_{j}$ that contains an empty bin.
By the capacity of bins and the discussion so far, $S_{h}$ satisfies the following conditions.
\begin{enumerate}
  \item There are $\rho + 1$ full red bins.
  \item There are $\beta + 2$ full blue bins.
  \item For each $i \in \{1,\dots,3m\} \setminus \{p,q,r\}$,
  the bin $b_{i}$ is a full red-blue bin that contains $a_{i}$ red units and $B - a_{i}$ blue units.
  \item There is one empty bin.
\end{enumerate}

Without loss of generality, assume that $\{p,q,r\} = \{3m-2, 3m-1, 3m\}$.
Let $S'$ be the instance of $\WSP$ obtained from $a_{1}, \dots, a_{3(m-1)}$ by the reduction above,
and $S'_{0}$ be the one obtained from $S'$ by adding $\rho + 1$ full red bins and $\beta + 2$ full blue bins.
Observe that $S'_{0}$ can be obtained from $S_{h}$ by renaming the bins.
Thus the sorting sequence from $S_{h}$ to $S_{\ell}$ can be applied to $S'_{0}$ by appropriately renaming the bins.
Therefore, we can apply the induction hypothesis to $S'_{0}$ and thus $a_{1}, \dots, a_{3(m-1)}$ is a yes-instance of \textsc{3-Partition}.
Since $\alpha = a_{3m-2} + a_{3m-1} + a_{3m} = B$, the instance $a_{1}, \dots, a_{3m}$ is also a yes-instance of \textsc{3-Partition}.
\end{proof}

\cref{cor:shortest-path-ball,cor:shortest-path-water,th:NP-h} imply that 
given an integer $t$ and a configuration $S$,
it is NP-complete to decide whether there is a sequence of length at most $t$
from $S$ to a sorted configuration under both settings in $\BSP$ and $\WSP$.
We observe a slightly stronger result for $\WSP$.

\begin{corollary}
\label{cor:water-short-npc}
Given an integer $t$ and an instance $S$ of $\WSP$,
it is NP-complete to decide whether there is a sorting sequence for $S$ with length at most $t$
even if $S$ is guaranteed to be a yes-instance.
\end{corollary}
\begin{proof}
From an instance $\langle a_{1}, \dots, a_{3m} \rangle$ of \textsc{3-Partition},
we first construct an instance of $\WSP$ as described in the proof of \cref{th:NP-h},
and then add a sufficient number of empty bins to guarantee that the resultant instance is a yes-instance.
This is always possible with a polynomial number of empty bins as we see in \cref{sec:bins}.\footnote{%
It is not difficult to show that this instance only needs a constant number of bins. In fact, two empty bins are enough to perform a greedy algorithm for sorting.}
Let $S$ denote the constructed instance. We set $t = 5m$

The proof of \cref{th:NP-h} implies that if $\langle a_{1}, \dots, a_{3m} \rangle$ is a yes-instance, 
then $S$ admits a sorting sequence of length $5m$. (See \figurename~\ref{fig:water}.)

Conversely, assume that $S$ admits a sorting sequence of length $5m$.
Recall that each of the $3m$ red-blue bins in $S$ contains $a_{i}$ red units at the top and $B - a_{i}$ blue units at the bottom for some $i$.
Since each full blue bin in the final configuration contains units from at least two original bins, we need at least one move for it.
Thus we need at least $2m$ moves to make $2m$ full blue bins.
This implies that we have at most $3m$ moves that involve red units.
Actually, the number of such moves is exactly $3m$
since each red unit has to move at least once.
Since $B/4 < a_{i} < B/2$ for all $i$,
each of the $m$ full red bins in the final configuration contains units from at least three original bins,
and thus it needs at least three moves.
If some red bin contains units from more than three original bins, then it needs at least four moves.
This contradicts the assumption that we have at most $3m$ moves for red units.
Thus we can conclude that each red bin in the final configuration contains units from exactly three original bins.
This gives a solution to the instance of \textsc{3-Partition} as the capacity of the bins is $B$.
\end{proof}

%% file: h2.tex
In this section, we focus on the special case of $h=2$ and $|C| = n$. 
In popular apps, it is often the case that $h$ is a small constant and $|C| = n$.
The case $h=2$ is the first nontrivial case in this setting.
Under this setting, every ball-move is a water-move and vice versa, except for moving (a) unit(s) from a bin with two units of the same color to an empty bin, which is a vacuous move.
Therefore, in this section we do not distinguish water-moves and ball-moves and simply call them moves.
We prove that in this setting, all instances with $k \ge 2$ are yes-instances.
We also show that we can find
a shortest sorting sequence in $O(n)$ time (if any exists).

We say that a bin of capacity $2$ is a \emph{full bin} if it contains two units,
and a \emph{half bin} if it contains one unit.

\begin{figure}\centering
	\begin{tikzpicture}[scale=0.4,shape=circle,inner sep=0pt,minimum size=5mm,>=stealth]
	\tikzset{C1/.style={fill=red!25}}
	\tikzset{C2/.style={fill=green!25}}
	\tikzset{C3/.style={fill=blue!25}}
	\tikzset{C4/.style={fill=red!45}}
	\tikzset{C5/.style={fill=red!45!green!45}}
	\tikzset{C6/.style={fill=green!45}}
	\tikzset{C7/.style={fill=blue!45}}
	\tikzset{C8/.style={fill=red!45!blue!45}}
	\tikzset{C9/.style={fill=red!30!green!50}}
	\tikzset{C10/.style={fill=green!50!blue!50}}
	\tikzset{C11/.style={fill=green!20!blue!50}}
	\tikzset{C12/.style={fill=red!20!blue!50}}
	\tikzset{C13/.style={fill=red!60!green!20}}
	\tikzset{C14/.style={fill=red!60!blue!20}}
	\tikzset{C15/.style={fill=black!30}}
	\tikzset{C16/.style={fill=black!10}}
\newcommand{\qbin}[4]{%
	\begin{scope}[xshift=#1cm]
	\foreach \cl [%
 		count = \k,
 		] in {#3} {
 		\filldraw[C\cl] (0,\k) -- (0,\k-1) -- (1,\k-1) -- (1,\k) -- cycle;
		\draw (0.5,\k - 0.5) node {\cl};
 	}
	\draw[thick] (-0.1,#4+0.3) -- (0,#4+0.2) -- (0,0) -- (1,0) -- (1,#4+0.2) -- (1.1,#4+0.3); 
	\end{scope}
}
		\draw (-5,0) node {$\vec{G}(S)$:};
		\foreach \i in {1,...,3} [
		\draw[rotate=120*\i-120] (0,2) node[draw,C\i] (a\i) {\i};
		]
		\begin{scope}[xshift=6.5cm]
		\foreach \i in {4,...,8} [%
		\draw[rotate=72*\i+72] (0,2) node[draw,C\i] (a\i) {\i};
		]
		\end{scope}
		\begin{scope}[xshift=13cm]
		\foreach \i in {9,...,12} [%
		\draw[rotate=90*\i-45] (0,2) node[draw,C\i] (a\i) {\i};
		]
		\end{scope}
		\draw (18,1.5) node [draw,C13] (a13) {13};
		\draw (18,-1.5) node [draw,C14] (a14) {14};
		\draw (21,-1) node [draw,C15] (a15) {15};
		\draw (24,1) node [draw,C16] (a16) {16};
		\draw[->,thick] (a1) -- (a2);
		\draw[->,thick] (a2) -- (a3);
		\draw[->,thick] (a3) -- (a1);
		\draw[->,thick] (a4) -- (a5);
		\draw[->,thick] (a4) -- (a8);
		\draw[->,thick] (a6) -- (a5);
		\draw[->,thick] (a6) -- (a7);
		\draw[->,thick] (a7) -- (a8);
		\draw[->,thick] (a10) -- (a9);
		\draw[->,thick] (a10) -- (a11);
		\draw[->,thick] (a12) -- (a11);
		\draw[->,thick,out=240] (a13) to [in=120] (a14);
		\draw[->,thick,out=300] (a13) to [in=60] (a14);
		\draw[->,thick,out=120] (a15) to [in=180] (21,1)  [out=0] to [in=60] (a15);
		\begin{scope}[xshift=-2cm,yshift=-6cm]
		\draw (-3,1) node {$S$:};
		\qbin{0.0}{0}{1,2}{2}
		\qbin{1.5}{0}{2,3}{2}
		\qbin{3.0}{0}{3,1}{2}
		\qbin{4.5}{0}{4,5}{2}
		\qbin{6.0}{0}{6,5}{2}
		\qbin{7.5}{0}{6,7}{2}
		\qbin{9.0}{0}{7,8}{2}
		\qbin{10.5}{0}{4,8}{2}
		\qbin{12.0}{0}{9}{2}
		\qbin{13.5}{0}{10,9}{2}
		\qbin{15.0}{0}{10,11}{2}
		\qbin{16.5}{0}{12,11}{2}
		\qbin{18.0}{0}{12}{2}
		\qbin{19.5}{0}{13,14}{2}
		\qbin{21.0}{0}{13,14}{2}
		\qbin{22.5}{0}{15,15}{2}
		\qbin{24.0}{0}{16}{2}
		\qbin{25.5}{0}{16}{2}
		\end{scope}		
	\end{tikzpicture}
\caption{Configuration $S$ and its graph representation $\vec{G}(S)$.}
\label{fig:h2}
\end{figure}
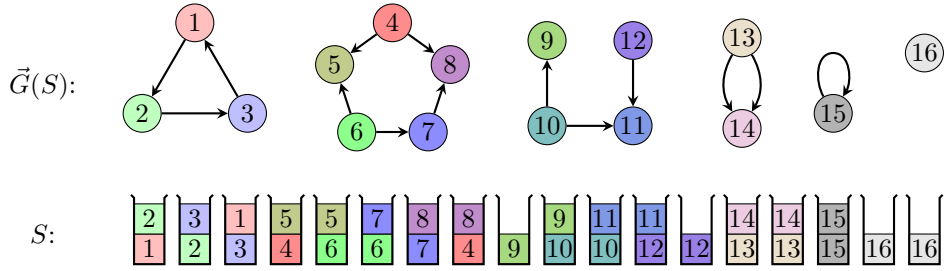
\begin{theorem}
\label{th:k=2}
If $h=2$, $|C|=n$, and $k \ge 2$, then all instances of\/ $\WSP$ are yes-instances.
Moreover, a shortest sorting sequence can be found in $O(n)$ time.
\end{theorem}
\begin{proof}
For the first claim of the theorem, it is enough to consider the case $k=2$.
Under a configuration $S$, we say a color $c \in C$ is \emph{sorted} if the two units of color $c$ are in the same bin.
For a configuration $S$, we define a directed multigraph $\vec{G}(S)=(V,A)$ as follows (see \figurename~\ref{fig:h2}).
The vertex set $V$ is the color set $C$.
We add one directed edge from $c$ to $c'$ for each full bin that contains a unit of color $c$ at the bottom and a unit of color $c'$ at the top. That is, $A$ consists of $S(b)$ for all full bins $b$.
(We may add self-loops here.)
For the directed multigraph $\vec{G}(S)$, we denote its underlying multigraph by $G(S)=(V,E)$, where $E$ is a multiset obtained from  $A$ by ignoring the directions.
Since $|C|=n$, each color appears twice in $S$.
When $S$ consists of full bins, $G(S)$ is a 2-regular graph with $n$ edges, 
which means that $G(S)$ is a set of cycles. We call a self-loop a \emph{trivial cycle}.
If $S$ also contains half bins, then $G(S)$ is a disjoint union of cycles and paths.

Let $p_S$ denote the number of nontrivial directed cycles in $\vec{G}(S)=(V,A)$,
$q_S$ the number of vertices of indegree 2 in $\vec{G}(S)=(V,A)$, 
and $r_S$ the number of vertices with no self-loop, i.e., the number of unsorted colors.
We prove that if $S$ has either two empty bins or one empty and two half bins,
then a shortest sorting sequence has length $p_S+q_S+r_S$.
Clearly the initial configuration, which has two empty bins, satisfies the condition.

We first prove that there exists a sorting sequence of length $p_S+q_S+r_S$.
We use an induction on $p_S+q_S+r_S$.
When $p_S+q_S+r_S=0$, 
all colors are sorted and thus $S$ is a goal.
Now we turn to the inductive step, which consists of four cases.

\noindent
(Case 1) Suppose $S$ has no half bins and $q_S=0$.
Note that when $S$ has no half bins, it has two empty bins.
Then one can move the top unit of an arbitrary bin to an empty bin. 
This eliminates a directed cycle in the graph. 
Then we have one empty and two half bins. The claim follows from the induction hypothesis.

\noindent
(Case 2) Suppose $S$ has no half bins and there is a vertex $c$ whose indegree is 2.
By moving the two units of $c$ to an empty bin, we use two moves,
while decreasing both $q_S$ and $r_S$. 
Then we have one empty and two half bins. The claim follows from the induction hypothesis.

\noindent
(Case 3) Suppose $S$ has two half bins one of which has a color of outdegree 0. 
Then, the other unit of that color is not below another unit in some bin.
By moving that unit on top of the half bin, we can sort the color by one move.
This decreases $r_S$ by one, and does not change $q_S$.
After this move, still we have two half bins.
The claim follows from the induction hypothesis.

\noindent
(Case 4) Otherwise, $S$ has two half bins $b$ and $b'$ and both colors in the half bins have outdegree 1.
Let $c_0$ be the color in $b$ and $(c_0, c_1, \ldots, c_m)$ the sequence of vertices such that $(c_{i-1}, c_{i})\in A$ for all $i$ where $c_m$ has outdegree 0.
By the assumption, $c_m$ is not the color of the unit of $b'$.
Therefore, $c_m$ has indegree 2, i.e., both units of $c_m$ appear as the top units of full bins.
We sort the color $c_m$ using the empty bin. 
It takes two moves and decreases both $r_S$ and $q_S$ by one. 
We then sort $c_{m-1},\ldots, c_0$ in this order by $m$ moves, which decreases $r_S$ by $m$, 
where we require no additional empty bins and finally the bin $b$ will be empty.
We have one empty and two half bins.
The claim follows from the induction hypothesis.

We next prove that there exists no sorting sequence of length less than $p_S+q_S+r_S$ with regardless of the number $k$ of empty bins.
Note that if all colors are sorted, $p_S+q_S+r_S=0$. 
We show that any move decreases the potential by at most one.

\noindent
(Case 1) If the color of the moved unit belongs to a directed cycle, this reduces $p_S$ by one but not $q_S$. 
The other unit of the same color has a unit of another color on its top. 
Therefore this cannot reduce $r_S$.

\noindent
(Case 2) If the color of the moved unit has indegree 2, 
this reduces $q_S$ by one but none of $p_S$ or $r_S$.

\noindent
(Case 3) Otherwise, any other kind of moves cannot reduce $p_S$ or $q_S$.
Clearly one move cannot reduce $r_S$ by two.

Since $p_S+q_S+r_S=O(n)$, we can sort any instance in $O(n)$ moves.
Each feasible move can be found in $O(1)$ time, which completes the proof.
%
\end{proof}

\begin{theorem}
\label{th:k=1}
If $h=2$, $|C|=n$, and $k=1$, then $\WSP$ can be solved in $O(n)$ time.
For a yes-instance, a shortest sorting sequence can be found in $O(n)$ time.
\end{theorem}
\begin{proof}
For a configuration $S$, we again use the directed multigraph $\vec{G}(S)=(V,A)$
and its underlying multigraph by $G(S)=(V,E)$ defined in the proof of \cref{th:k=2}.
Let $S_0$ be a given instance of $\WSP$ with $h=2$, $k=1$, and $|C|=n$.
To simplify, we assume that $S_0$ has no full bin that contains two units of the same color without loss of generality.
(Hereafter, when we have a full bin that contains two units of the same color, we remove it from the configuration and never touch.
In terms of $\vec{G}(S)$, $\vec{G}(S)$ has no self-loop, 
and once $\vec{G}(S)$ produces a vertex with a self-loop, we remove it from $\vec{G}$.)
Observe that if a configuration $S$ consists of $n'$ full bins and one empty bin,
then $G(S)$ is a 2-regular multigraph with $n'$ edges.
That is, $G(S)$ is a set of cycles of size at least 2 since we remove vertices with self-loops.

Now we focus on a cycle $C'$ in $G$. When $C'$ is a directed cycle in $\vec{G}$ (as $C_3$ in \figurename~\ref{fig:h2}),
it is easy to remove the vertices in $C'$ from $G$ by using one empty bin.
When $C'$ contains exactly one vertex $v$ of outdegree $2$ in $\vec{G}$, $C'$ contains another vertex $u$ of
indegree $2$. In this case, we first move two units of color $u$ to the empty bin, and 
we can sort the other bins using two half bins. Lastly, we can get two units of color $v$ together,
and obtain an empty bin again.
Now we consider the case that $C'$ contains (at least) two vertices $v$ and $v'$ of outdegree $2$ in $\vec{G}$ 
(as $C_5$ in \figurename~\ref{fig:h2}).
In this case, we can observe that we eventually get stuck with two half bins of colors $v$ and $v'$.
That is, $S$ is a no-instance if and only if $\vec{G}(S)$ contains at least one cycle $C'$ that
has at least two vertices of outdegree $2$ in $\vec{G}(S)$.

In summary, we can conclude that the original instance $S_0$ can be sorted if and only if 
every cycle in $G(S_0)$ contains at most one vertex of outdegree $2$ in $\vec{G}(S_0)$.
The construction of $\vec{G}(S_0)$ and $G(S_0)$ can be done in $O(n)$ time,
and this condition can be checked in $O(n)$ time.
Moreover, it is easy to observe that sorting items in each cycle of length $n'$ in $G(S_0)$ requires $n'+1$ moves,
and $n'+1$ moves are sufficient.
Therefore, the optimal sorting requires $n+\ell$ moves, 
where $\ell$ is the number of cycles in $G(S_0)$, which can be computed in $O(n)$ time.
\end{proof}

By the proofs of Theorems \ref{th:k=2} and \ref{th:k=1}, we have the following corollary:
\begin{corollary}
\label{cor:k<=2}
If $h=2$ and $|C|=n$, any yes-instance of\/ $\WSP$ has a solution of length $O(n)$.
\end{corollary}

%% file: bins.tex
In this section, we consider the minimum number $k(n,h)$ for $n$ and $h$
such that all instances of $\WSP$ with $n$ full bins of capacity $h$ with $k(n,h)$ empty bins are yes-instances.
We can easily see that such a number exists and that $k(n,h)\le n$ as follows.
Let $S$ be an instance with $n$ full and $k$ empty bins of capacity $h$ such that $k \ge n$.
For each color $c$ used in $S$, let $j_{c}$ be the positive integer such that $S$ contains $j_{c} \cdot h$ units of color $c$.
We assign $j_{c}$ empty bins to each color $c$.
Since $\sum_{c} j_{c} = n$, we can easily ``bucket sort'' $S$ by moving water units in the $n$ bins to empty bins
following the assignment of the empty bins to the colors.
On the other hand, as shown in \cref{th:k=1}, we have no-instances when $k=1$ even if $h=2$ 
(see \cref{sec:conc} for a larger no-instance).

In the following, we improve the upper bound of $k(n,h)$ and show a lower bound.

\newcommand{\heightis}[4]{
	\draw[<-] (#1,#2) -- (#1,0.5*#2+0.5*#3-0.5);
	\draw[<-] (#1,#3) -- (#1,0.5*#2+0.5*#3+0.5);
	\draw (#1, 0.5*#2+0.5*#3) node {#4};
}
\newcommand{\widthis}[4]{
	\draw[<-] (#1,#3) -- (0.5*#1+0.5*#2-0.45,#3);
	\draw[<-] (#2,#3) -- (0.5*#1+0.5*#2+0.45,#3);
	\draw (0.5*#1+0.5*#2, #3) node {#4};
}
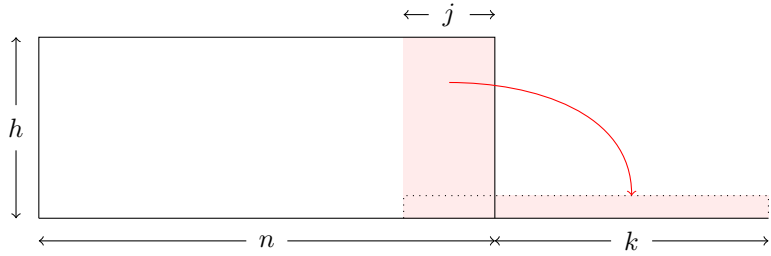
\begin{figure}\centering
 \begin{tikzpicture}[scale=0.6,xscale=1]
 	\fill[red!8] (8,0) -- (8,4) -- (10,4) -- (10,0);
 	\fill[red!8] (8,0) -- (8,0.5) -- (16,0.5) -- (16,0);
 	\draw (0,0) -- (0,4) -- (10,4) -- (10,0);
 	\draw (0,0) -- (16,0);
 	\draw[style=dotted] (8,0) -- (8,0.5) -- (16,0.5) -- (16,0);
 	\draw[color=red, out=0, in=90, ->] (9,3) to (13,0.5); 
 	\heightis{-0.5}{0}{4}{$h$}
 	\widthis{0}{10}{-0.5}{$n$}
 	\widthis{10}{16}{-0.5}{$k$}
 	\widthis{8}{10}{4.5}{$j$}
 \end{tikzpicture}
\caption{Refinement of the naive bucket sort.} 
\label{fig:sufficient}
\end{figure}

\begin{theorem}
$k(n,h) \le \ceil{\frac{h-1}{h}n}$.
\end{theorem}
\begin{proof}
It suffices to show the lemma for the case where $n=|C|$.
We refine the bucket-sort base algorithm described above.
See \cref{fig:sufficient}.
Suppose that we have at least $k=\ceil{\frac{h-1}{h}n}$ empty bins.
Then, using those empty bins, one can make $j=\floor{\frac{n}{h}}$ other bins monochrome, since $k \ge (h-1)j$.
Now we have $k+j=n$ monochrome bins.
When some of those bins have the same color, we merge them.
Then we can dedicate different $n$ bins to pairwise different colors,
and we can use them to sort the remaining bins by the bucket sort.
\end{proof}
We note that $\ceil{\frac{h-1}{h}n}=n$ when $h>n$, which coincides with the naive bucket sort algorithm.

\begin{theorem}\label{thm:sufbin_lbound}
$k(n,h) \ge \ceil{\frac{19}{64} \min\{n,h\}}$.

\end{theorem}
\begin{proof}
We show a configuration of $h=n$ bins where every bin has the same contents {$(1,\dots,n)$}. 
In the case where $n > h$, one can add $n-h$ extra bins whose contents are already sorted.
In the case where $n < h$, one can arbitrarily pad the non-empty bins of the above configuration with $h-n$ extra units of each color at the bottom.

\begin{figure}\centering
\begin{tikzpicture}[scale=0.25,yscale=-1.1]
	\newcommand{\Lshp}[4]{%
		\filldraw[fill=black!8, draw=black] (0,0) -- (#1-1,0) -- (#1-1,#3) -- (#1,#3) -- (#1,#2) -- (0,#2) -- (0,0);
		\draw (#1-1,0) -- (#1,0);
		\draw (2,#2-1) node {#4};
	}
	\newcommand{\Lshape}[5]{%
		\filldraw[fill=black!8, draw=black] (0,#2) -- (#1,#2) -- (#1,0) -- (#1+#3-1,0) -- (#1+#3-1,#4) -- (#1+#3,#4) -- (#1+#3,#2+1) -- (0,#2+1) -- (0,#2);
		\draw (#1+#3-1,0) -- (#1+#3,0) -- (#1+#3,#4);
		\draw (2,#2+0.5) node {#5};
	}
	\newcommand{\Fullbin}[2]{%
		\filldraw[fill=black!8, draw=black] (#1,0) -- (#1+1,0) -- (#1+1,25) -- (#1,25) -- (#1,0);
		\draw (#1+0.5,24) node {#2};
	}
	\draw[style=densely dashed] (40,25) -- (0,25) -- (0,0) -- (40,0);
	\draw (36,0) -- (37,0);
	\draw (36,25) -- (37,25);
	\Lshp{4}{9}{5}{450}
	\Lshape{4}{9}{4}{6}{449}
	\Lshape{11}{12}{4}{11}{441}
	\Lshape{15}{13}{3}{1}{440}
	\Lshape{21}{16}{3}{15}{401}
	\Lshape{24}{17}{2}{1}{400}
	\Lshape{29}{20}{2}{19}{321}
	\Lshape{31}{21}{1}{1}{320}
	\Lshape{35}{24}{1}{23}{1}
	\Fullbin{37}{451}
	\Fullbin{40}{640}
	\widthis{0}{4}{-1}{4}
	\widthis{4}{8}{-1}{4}
	\widthis{11}{15}{-1}{4}
	\widthis{15}{18}{-1}{3}
	\widthis{21}{24}{-1}{3}
	\widthis{24}{26}{-1}{2}
	\widthis{29}{31}{-1}{2}
	\widthis{31}{32}{-1}{1}
	\widthis{35}{36}{-1}{1}
	\widthis{36}{37}{-1}{1}
	\widthis{37}{38}{-1}{1}
	\widthis{40}{41}{-1}{1}
	\heightis{-1}{0}{9}{$191$}
	\heightis{-2}{0}{25}{$640$}
	\draw[<-] (37,25.5) -- (38,25.5);
	\draw[<-] (41,25.5) -- (40,25.5);
	\draw (39, 25.5) node {190};
\end{tikzpicture}
\caption{When the first unit of color $1$ has been moved, where $n=640$ and $k=190$.} 
\label{fig:sufbin_lbound}
\end{figure}
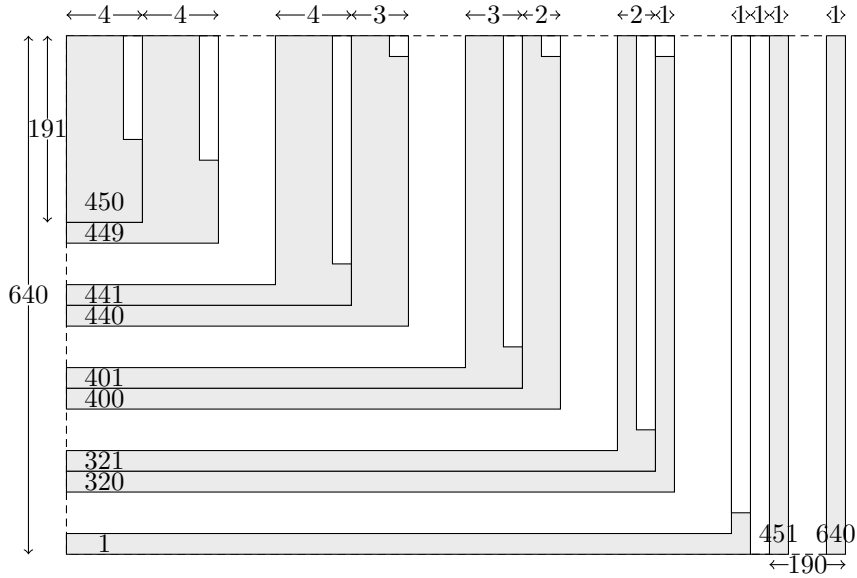

Suppose that it is solvable with $k$ empty bins.
We consider the very first moment when a unit of color $1$ has been moved.
Figure~\ref{fig:sufbin_lbound} shows an example configuration with $n=640$ and $k=190$.
Note that any initially empty bin will always be monochrome.
Let $U$ be the set of colors $c$ such that at least one unit of the color $c$ occupies a bin which was initially empty, where $|U| \le k$,
and $r_c$ be the number of units of color $c$ which have been removed from the initial location.
Particularly $r_1=1$.
Then we have $r_{c+1} \ge r_{c}$ for all colors $c \ge 1$, since a unit of color $c$ can be removed only after all the units above have been removed.
Moreover, if $c \notin U$, units of color $c$ which can be put only on units of the same color.
That is, from both the source and target bins involved in the move, the unit of color $c+1$ must have been removed.
Each of the target bins accepts at most $(n-c)$ units of water color $c$ on top of the unit initially located there.
Therefore, for $c \notin U$, $r_c \le (r_{c+1}-r_c)(n-c)$ (see Figure~\ref{fig:sufbin_ineq}), i.e.,
\begin{equation}\label{eq:sufbin_lbound}
	r_{c+1}-r_c \ge \lrceil{\frac{r_c}{n-c}}\,.
\end{equation}
We will give a lower bound of the size of $U$ that admits a sequence of integers $(r_1,\dots,r_n)$ with $1 \le r_1 \le \dots \le r_n \le n$ that satisfies Eq.~(\ref{eq:sufbin_lbound}) for $c \notin U$.
Suppose that $(r_1,\dots,r_n)$ satisfies the condition with  $c \in U$ and $c_{c+1} \notin U$. 
Then, one can easily see that $U'=U\setminus\{c\}\cup\{c+1\}$ admits a sequence $(r_1,\dots,r_{c-1},r_{c}',r_{c+1},r_n)$ with $r_c'=r_{c+1}$ that satisfies Eq.~(\ref{eq:sufbin_lbound}).
Therefore, we may and will assume that $U=\{n-k+1,\dots,n\}$.
\begin{figure}\centering
 \begin{tikzpicture}[scale=0.6,xscale=1]
 	\fill[red!8] (8,0.5) -- (8,4) -- (10,4) -- (10,0.5);
 	\fill[red!8] (10,0) -- (10,0.5) -- (16,0.5) -- (16,0);
 	\draw (4,0) -- (4,4) -- (8,4) -- (8,0.5);
 	\draw (4,0) -- (16,0) -- (16,0.5) -- (4,0.5);
 	\draw (8,4) -- (10,4) -- (10,0);
 	\draw (4,1) -- (8,1);
 	\draw[color=red, out=0, in=90, <-] (9.5,3.5) to (13,0.5); 
 	\heightis{9}{0.5}{4}{$n-c$}
 	\widthis{8}{16}{-1.0}{$r_{c+1}$}
 	\widthis{10}{16}{-0.5}{$r_c$}
	\draw (5, 0.75) node {$c+1$};
	\draw (5, 0.25) node {$c$};
 \end{tikzpicture}
\caption{If $c \notin U$, $r_c \le (r_{c+1}-r_c)(n-c)$.} 
\label{fig:sufbin_ineq}
\end{figure}
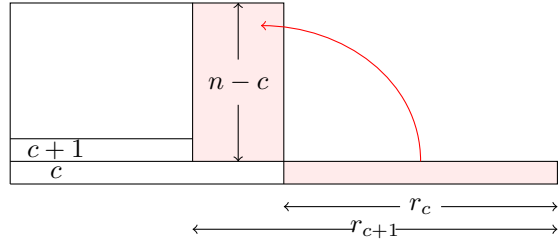

For $c \le n-k$, $r_1 = 1$ and $r_{c+1} - r_c \ge 1$ implies $r_c \ge c$.
Hence, for $n/2 < c \le n-k$,
\[
	r_{c+1}-r_c \ge \lrceil{\frac{r_c}{n-c}} > \frac{n/2}{n-n/2} = 1
\]
implies $r_c \ge 2c - n/2$.
Hence, for $(5/8)n < c \le n-k$,
\[
	r_{c+1}-r_c \ge \lrceil{\frac{r_c}{n-c}} > \frac{(5/4)n-n/2}{n-(5/8)n} = 2
\]
implies $r_c \ge 3c - (9/8)n$.
Hence, for $(11/16)n < c \le n-k$,
\[
	r_{c+1}-r_c \ge \lrceil{\frac{r_c}{n-c}} > \frac{(33/16)n-(9/8)n}{n-(11/16)n} = 3
\]
implies $r_c \ge 4c - (29/16)n$.
On the other hand, $r_c \le n$.
That is, $c \le n-k$ implies $c \le (45/64)n$, i.e., $k \ge (19/64)n$.
\end{proof}

%% file: conc.tex
In this paper, we investigate the problems of solvability of ball and water sort puzzles.
We show that both problems are equivalent and NP-complete.
We also show that even for trivial instances of the water sort puzzle,
it is NP-complete to find a shortest sorting sequence.

\begin{figure}\centering
	\newcommand{\qbin}[4]{%
	\tikzset{C1/.style={fill=red!15}}
	\tikzset{C2/.style={fill=green!15}}
	\tikzset{C3/.style={fill=blue!15}}
	\tikzset{C4/.style={fill=red!45}}
	\tikzset{C5/.style={fill=green!45}}
	\tikzset{C6/.style={fill=blue!45}}
	\tikzset{C7/.style={fill=red!50!green!50}}
	\tikzset{C8/.style={fill=green!50!blue!50}}
	\tikzset{C9/.style={fill=blue!50!red!50}}
	\begin{scope}[xshift=#1cm]
	\foreach \cl [%
 		count = \k,
 		] in {#3} {
 		\filldraw[C\cl] (0,\k) -- (0,\k-1) -- (1,\k-1) -- (1,\k) -- cycle;
		\draw (0.5,\k - 0.5) node {\cl};
 	}
	\draw[thick] (-0.1,#4+0.3) -- (0,#4+0.2) -- (0,0) -- (1,0) -- (1,#4+0.2) -- (1.1,#4+0.3); 
	\end{scope}
}
	\begin{tikzpicture}[scale=0.4]
		\qbin{0.0}{0}{7,4,1}{3}
		\qbin{1.5}{0}{8,4,1}{3}
		\qbin{3.0}{0}{9,4,1}{3}
		\qbin{4.5}{0}{7,5,2}{3}
		\qbin{6.0}{0}{8,5,2}{3}
		\qbin{7.5}{0}{9,5,2}{3}
		\qbin{9.0}{0}{7,6,3}{3}
		\qbin{10.5}{0}{8,6,3}{3}
		\qbin{12.0}{0}{9,6,3}{3}
		\qbin{13.5}{0}{}{3}
		\qbin{15}{0}{}{3}
		\draw (8,-1) node {(a)};
	\end{tikzpicture}
\\
\renewcommand{\qbin}[4]{%
	\tikzset{C1/.style={fill=red!30}}
	\tikzset{C2/.style={fill=red!50}}
	\tikzset{C3/.style={fill=red!70}}
	\tikzset{C4/.style={fill=green!15}}
	\tikzset{C5/.style={fill=green!30}}
	\tikzset{C6/.style={fill=green!45}}
	\tikzset{C7/.style={fill=green!60}}
	\tikzset{C8/.style={fill=green!75}}
	\tikzset{C9/.style={fill=green!90}}
	\tikzset{C10/.style={fill=blue!30}}
	\tikzset{C11/.style={fill=blue!50}}
	\tikzset{C12/.style={fill=blue!70}}
	\begin{scope}[xshift=#1cm]
	\foreach \cl [%
 		count = \k,
 		] in {#3} {
 		\filldraw[C\cl] (0,\k) -- (0,\k-1) -- (1,\k-1) -- (1,\k) -- cycle;
		\draw (0.5,\k - 0.5) node {\cl};
 	}
	\draw[thick] (-0.1,#4+0.3) -- (0,#4+0.2) -- (0,0) -- (1,0) -- (1,#4+0.2) -- (1.1,#4+0.3); 
	\end{scope}
}
	\begin{tikzpicture}[scale=0.4]
		\qbin{0.0}{0}{10,4,4,1}{4}
		\qbin{1.5}{0}{10,5,5,1}{4}
		\qbin{3.0}{0}{11,6,6,1}{4}
		\qbin{4.5}{0}{11,7,7,1}{4}
		\qbin{6.0}{0}{11,8,8,2}{4}
		\qbin{7.5}{0}{11,9,9,2}{4}
		\qbin{9.0}{0}{12,4,4,2}{4}
		\qbin{10.5}{0}{12,5,5,2}{4}
		\qbin{12.0}{0}{12,6,6,3}{4}
		\qbin{13.5}{0}{12,7,7,3}{4}
		\qbin{15.0}{0}{10,8,8,3}{4}
		\qbin{16.5}{0}{10,9,9,3}{4}
		\qbin{18.0}{0}{}{4}
		\qbin{19.5}{0}{}{4}
		\qbin{21.0}{0}{}{4}
		\draw (11,-1) node {(b)};
	\end{tikzpicture}
\caption{No-instances for (a) $h=3$, $k=2$, $n=9$ and (b) $h=4$, $k=3$, $n=12$.}
\label{fig:no_k2}
\end{figure}
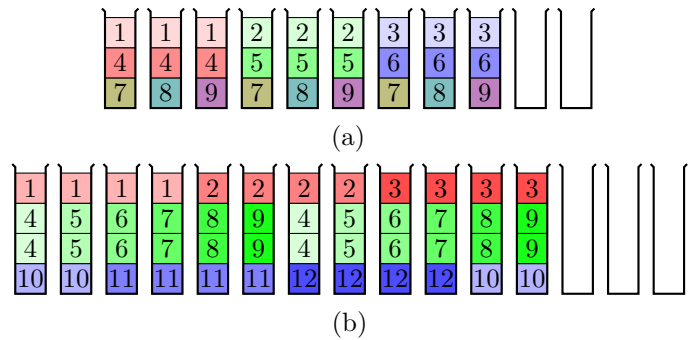

As discussed in Section \ref{sec:bins}, any instance can be sorted when $k$ is large enough.
Especially, as discussed in Section \ref{sec:h2}, $k=2$ is enough for $h=2$, while we have a no-instance for $k=1$ and $h=2$.
On the other hand, there exist no-instances with $k=2$ for $h=3$ and $k=3$ for $h=4$ (\figurename~\ref{fig:no_k2}).
We constructed them by hand and verified that they are no-instances by using a computer program.
It is interesting to give the boundary of $k$ for a given $h$, especially, $h=3$, that allows us to sort any input.
Does $k$ depend on both $h$ and $n$, or is it independent from $n$?


Recently, a Japanese puzzle maker produces a commercial product of the ball sort puzzle.
\footnote{\url{https://www.hanayamatoys.co.jp/product/products/product-cat/puzzle/}}
They extend the ball sort puzzle by introducing three different aspects as follows.
(1) It contains three different capacities of bins. That is, it allows to use different size of bins.
(2) The balls are not only colored, but also numbered. In some problems, the puzzle asks us to 
not only arrange the balls of the same color into a bin, but also they should be in order in each bin.
In some other problems, the puzzle asks us to arrange the balls 
so that the sums of the numbers in every non-empty bins are equal. 
(3) It contains transparent balls. That is, these transparent balls are just used for arranging balls,
and their final positions do not matter in the final state.
Those extensions seems to be reasonable future work.